\documentclass[conference]{IEEEtran}
\IEEEoverridecommandlockouts

\usepackage{cite}
\usepackage{amsmath,amssymb,amsfonts,amsthm}
\usepackage{algorithmic}
\usepackage{graphicx}
\graphicspath{{figs/}}
\usepackage{textcomp}
\usepackage{xcolor}
\usepackage{subcaption}
\usepackage{float}

\newtheorem{theorem}{Theorem}
\newtheorem{remark}{Remark}

\usepackage{amsmath}
\usepackage{booktabs}
\DeclareMathOperator{\sech}{sech}

\def\BibTeX{{\rm B\kern-.05em{\sc i\kern-.025em b}\kern-.08em
		T\kern-.1667em\lower.7ex\hbox{E}\kern-.125emX}}
\begin{document}
	
	\title{Direct Vector Field Guidance for  Path Following}
	
	\author{
		\IEEEauthorblockN{
			\begin{tabular}{cc}
				\begin{minipage}{0.45\textwidth}
					\centering
					Amit Shivam\\
					\textit{SYSTEC, ISR--Porto, ARISE-LA}\\
					\textit{Universidade do Porto}\\
					Porto, Portugal\\
					amitshivam@upwind.pt
				\end{minipage}
				&
				\begin{minipage}{0.45\textwidth}
					\centering
					Nikhil Kumar Singh\\
					Department of Aviation Science\\
					Rajiv Gandhi National Aviation University\\
					Amethi, Uttar Pradesh, India \\
					nikhil.singh@rgnau.ac.in
				\end{minipage}
				\\[3.0em]
				\begin{minipage}{0.45\textwidth}
					\centering
					Luís~Tiago~Paiva\\
					\textit{SYSTEC, ISR--Porto, ARISE-LA}\\
					\textit{Faculdade de Engenharia}\\
					\textit{Universidade do Porto}\\
					Porto, Portugal\\
					ltpaiva@fe.up.pt
				\end{minipage}
				&
				\begin{minipage}{0.45\textwidth}
					\centering
					Fernando~A.C.C.~Fontes\\
					\textit{SYSTEC, ISR--Porto, ARISE-LA}\\
					\textit{Faculdade de Engenharia}\\
					\textit{Universidade do Porto}\\
					Porto, Portugal\\
					faf@fe.up.pt
				\end{minipage}
			\end{tabular}
		}
	}
	
	
	\maketitle
	\begingroup
	\renewcommand\thefootnote{}
	\footnotetext{This work is financially supported by Project 2022.02801.PTDC-UPWIND-ATOL (https://doi.org/10.54499/2022.02801.PTDC), as well as by grants 2021.06313.BD and 2021.07346.BD, and by the SYSTEC – Research Centre for Systems and Technologies (UID/00147) and the Associate Laboratory ARISE – Advanced Production and Intelligent Systems (LA/P/0112/2020, DOI: 10.54499/LA/P/0112/2020), all funded by Fundação para a Ciência e a Tecnologia, I.P./MECI through national funds.}
	\endgroup

	\begin{abstract}
		This paper presents a path-following guidance framework for a turn-rate-constrained unmanned aerial vehicle (UAVs). The proposed guidance strategy uses commanded course angle as a tanh-hyperbolic shaping function of the lateral deviation from the desired path.
		Stability characteristics illustrating asymptotic convergence to the desired path are established using Lyapunov analysis.  Analytical expressions for the turn-rate and corresponding maximum turn-rate bound are derived, enabling direct gain selection for a turn-rate-limited vehicle. Numerical simulations are performed for straight-line following scenarios under both ideal and non-ideal initial heading conditions. Furthermore, a comparative study against representative arctan-, arcsine-, and arccos-based vector-field guidance laws is conducted using analytically feasible gains selected from their respective turn-rate bounds. Comparative numerical results demonstrate  the proposed formulation achieves reliable boundary convergence while maintaining smooth steering behavior and a compact analytical gain-design framework. 
	\end{abstract}
	
	
		
		
	
	\section{Introduction}
	
	Autonomous path-following capability has become an essential requirement for UAVs operating in geo-fencing, perimeter surveillance, infrastructure inspection, and persistent monitoring missions.
	In such applications, the UAV is required to converge toward and follow a prescribed path while maintaining smooth maneuvering and satisfying the vehicle's dynamic constraints.
	For fixed-wing UAVs in particular, the achievable maneuverability is constrained by the maximum allowable turn-rate and minimum turning radius.
	Consequently, guidance laws that demand excessive transient steering commands may produce infeasible trajectories and degrade tracking performance.
	This work develops analytically tractable, turn-rate-aware guidance laws for path-following.
	
	Among the various path-following approaches available in the literature\cite{sujit2014unmanned}, vector-field guidance methods\cite{frew2008coordinated,goncalves2009artificial,pothen2017curvature,kapitanyuk2017,nelson2007vector,amitRatnoo2021,shivamRatnoo2023,amit2022vector,shivamLC2024,shivam2025robustvf,shivam2024VFclass,pal2026curvature} have emerged as attractive solutions due to their smooth convergence characteristics, low computational complexity, and ease of implementation.
	Existing vector-field guidance approaches can broadly be categorized into gradient-based vector fields\cite{frew2008coordinated,goncalves2009artificial,pothen2017curvature,kapitanyuk2017} and direct vector fields\cite{nelson2007vector, amitRatnoo2021,shivamRatnoo2023,amit2022vector,shivamLC2024,pal2026curvature}.
	Gradient-based vector fields generally construct the heading guidance command using gradients of potential functions or implicit path representations.
	Although these methods provide strong geometric interpretation and convergence guarantees, the computation of gradient and curvature computations can be computationally demanding, complicating real-time implementation and maneuvering-feasibility analysis.
	Direct vector-field guidance laws offer an alternative computationally inexpensive formulation in which the commanded course angle is directly generated as a nonlinear function of the UAV position relative to the desired path.
	In~\cite{nelson2007vector} is proposed one of the earliest direct vector-field guidance laws using an arctangent shaping profile for path following.
	Subsequently, arcsine-based vector fields were developed in \cite{shivamRatnoo2023,amit2022vector,shivamLC2024,shivam2025robustvf,shivam2024VFclass}, yielding reduced peak curvature demands compared to \cite{nelson2007vector}.
	More recently, exponential-arccosine shaping profiles \cite{pal2026curvature} have been proposed to improve convergence smoothness and reduce peak curvature through exponentially modulated heading commands.
	These direct vector-field guidance laws provide smooth steering behavior while avoiding the complexity associated with gradient-based guidance frameworks.
	
	Despite these works, the maneuvering-feasibility characteristics of direct vector fields remain largely unexplored.
	In particular, aggressive convergence profiles often produce large transient turn-rate demands near the path, potentially exceeding the allowable steering capability of turn-constrained UAVs.
	Motivated by these observations, this paper presents a turn-rate-limited analytical framework for direct vector-field guidance laws that preserves smooth saturation and strong convergence properties while enabling explicit analytical characterisation of turn-rate and direct gain selection from prescribed vehicle manoeuvring limits. 
	
	The main contributions of this paper are summarized as follows:
	\begin{enumerate}
		\item development of a  direct vector-field guidance law for turn-limited UAV path following,
		\item analytical derivation of turn-rate profiles and corresponding maximum turn-rate bounds for the proposed method with explicit gain-selection conditions directly linking the shaping gain to prescribed vehicle turn-rate limitations, 
		\item comparative analytical investigation of arctangent-, arcsine-, and arccosine-based vector fields to characterise their turn-rate and convergence properties,
		\item simulation studies under ideal and non-ideal heading conditions to validate maneuvering feasibility and tracking performance.
	\end{enumerate}
	
	The remainder of the paper is organized as follows. Section II presents the problem formulation for turn-rate-limited UAV path following. Section III develops the proposed guidance law. Comparative analytical investigations with existing vector-field guidance laws are presented in Section IV.
	Section V provides simulation results, and Section VI presents concluding remarks.

	\section{Problem Statement}
	
	Consider a path-following scenario shown in Fig.~\ref{fig: path_following}, where a fixed-wing UAV is required to follow a prescribed path while satisfying the turn-rate limitations of the vehicle. The desired path is initially considered as a straight-line path aligned with the $Y$-axis. The UAV position in the inertial frame is denoted by $(x,y)$, where $x$ represents the lateral deviation from the desired path. The UAV ground speed and course angle are represented by $V_g$ and $\chi$, respectively. The maximum allowable turn-rate of the UAV is denoted by $\omega_{\max}$.
	\begin{figure}[h]
		\centering
		\includegraphics[width=0.9\linewidth]{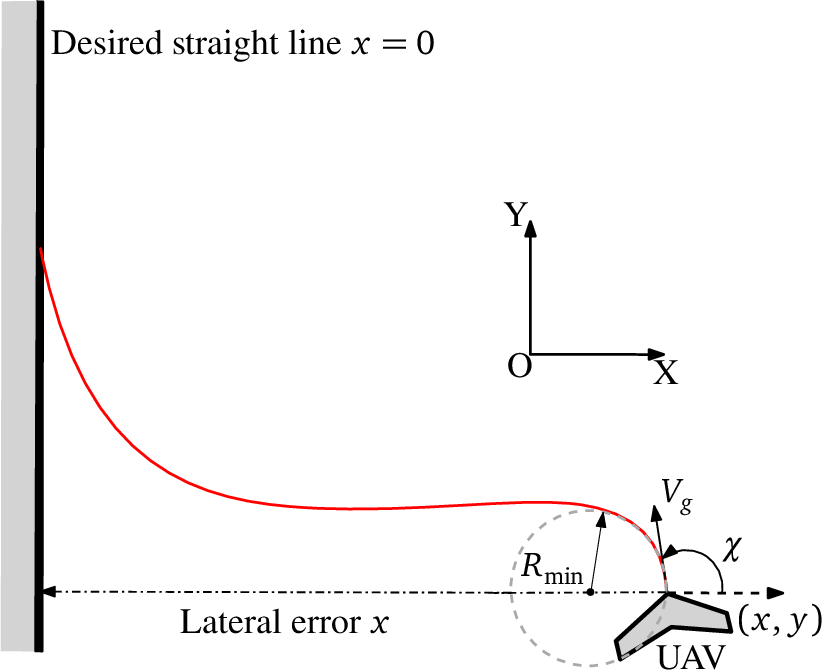}
		\caption{Path-following scenario.}
		\label{fig: path_following}
	\end{figure}
	The planar kinematic model of the UAV is expressed as
	\begin{align}
		\dot{x} = V_g \cos \chi\,, \qquad
		\dot{y} = V_g \sin \chi
		\label{eq: kinematic model}
	\end{align}
	The first-order course dynamics of the UAV are modeled as
	\begin{equation*}
		\dot{\chi} = \alpha(\chi_d-\chi),
	\end{equation*}
	where $\chi_d$ is the commanded course angle generated by the vector-field guidance law and $\alpha>0$ is the controller gain. The objective is to design a direct vector-field guidance law that drives the UAV toward the desired path while ensuring smooth convergence and maintaining turn-feasible motion throughout the transient maneuver.
	Mathematically, the guidance objective is to design $\chi_d$ such that
	\begin{align*}
		x(t) \rightarrow 0, \qquad 
		\chi(t) \rightarrow \frac{\pi}{2},
	\end{align*}
	while satisfying the turn-rate feasibility condition
	\begin{equation*}
		|\dot{\chi}_d| \leq \omega_{\max},
	\end{equation*}
	where $\dot{\chi}_d$ denotes the commanded turn-rate of the vector field guidance. The emphasis of this work is therefore on developing a feasible direct vector-field guidance law whose shaping parameter can be analytically related to the maximum allowable turn-rate limitation of the UAV.
	
	\section{Proposed Vector Field Guidance Method}
	
	This section develops a direct vector-field guidance law for path following. 
	
	\subsection{Guidance Law}
	
	The proposed course-angle command is defined as
	\begin{equation}
		\chi_d(x) = \frac{\pi}{2} + \sin^{-1}\left(\tanh(k_\mathrm{atanh}x)\right)
		\label{eq:proposed_chi}
	\end{equation}
	where $k_\mathrm{atanh}>0$ is the vector-field shaping gain. The term $\pi/2$ aligns the desired motion with the positive $Y$-axis when the UAV lies on the path and can be rotated for the inclined path.
	From \eqref{eq:proposed_chi}, when $x=0$, $\chi_d(0)=\frac{\pi}{2}$, which corresponds to motion along the desired path. For $x>0$, $\tanh(k_\mathrm{atanh}x)>0$, and hence $\chi_d>\pi/2$. Since the course angle is measured from the positive $X$-axis in the kinematic model \eqref{eq: kinematic model}
	the condition $\chi_d>\pi/2$ implies $\cos(\chi_d)<0$, which drives the UAV toward decreasing $x$. Similarly, for $x<0$, $\tanh(k_\mathrm{atanh}x)<0$, and hence $\chi_d<\pi/2$, which implies $\cos(\chi_d)>0$ and drives the UAV toward increasing $x$. Therefore, the vector field always points toward the desired path $x=0$.
	Using $\cos\left({\pi}/{2}+\theta\right)=-\sin(\theta)$ and setting $\theta=\sin^{-1}(\tanh(k_\mathrm{atanh}x))$, the lateral error dynamics under ideal course following become
	\begin{equation}
		\dot{x}
		=
		V_g\cos(\chi_d)
		=
		-V_g\tanh(k_\mathrm{atanh}x).
		\label{eq:lateral_dynamics}
	\end{equation}
	Equation~\eqref{eq:lateral_dynamics} shows that the lateral motion is directed toward the path for all $x\neq 0$. Fig.~\ref{fig: Results for course angle and the corresponding vector field} illustrates the resulting course-angle profile and vector field plots. The continuous heading transition across the entire lateral-error domain is shown in Fig. \ref{fig: Commanded course angle} and the corresponding vector fields are plotted in Fig. \ref{fig: Vector field}.
	
	\begin{figure}
		\begin{subfigure}[b]{.25\textwidth}
			\centering			\includegraphics[trim={2pt 3pt 2pt 3pt}, clip, width=\linewidth]{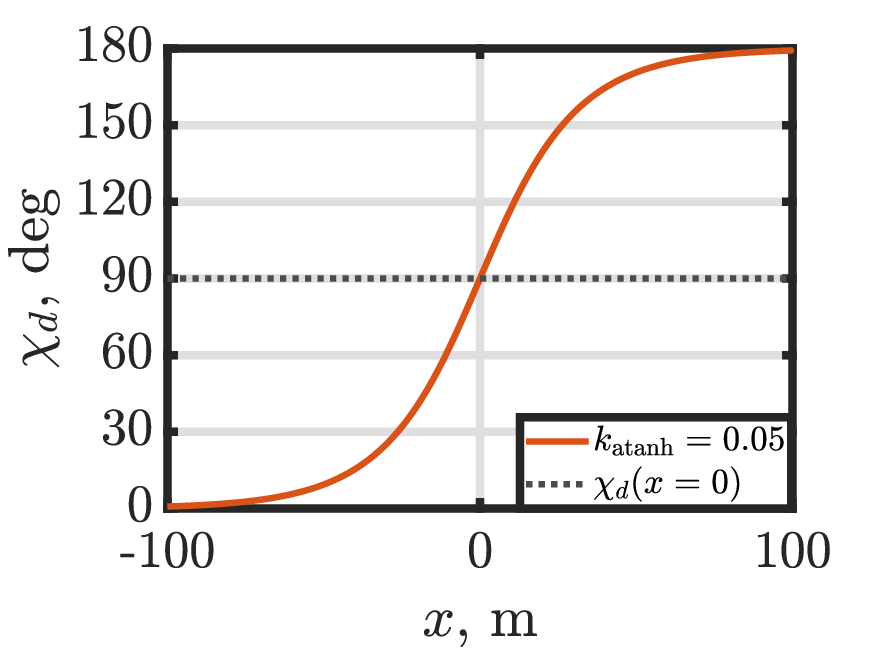 }
			\caption{ Commanded course angle}
			\label{fig: Commanded course angle}     
		\end{subfigure}%
		\begin{subfigure}[b]{.25\textwidth}
			\centering		\includegraphics[trim={2pt 3pt 2pt 3pt}, clip, width=\linewidth]{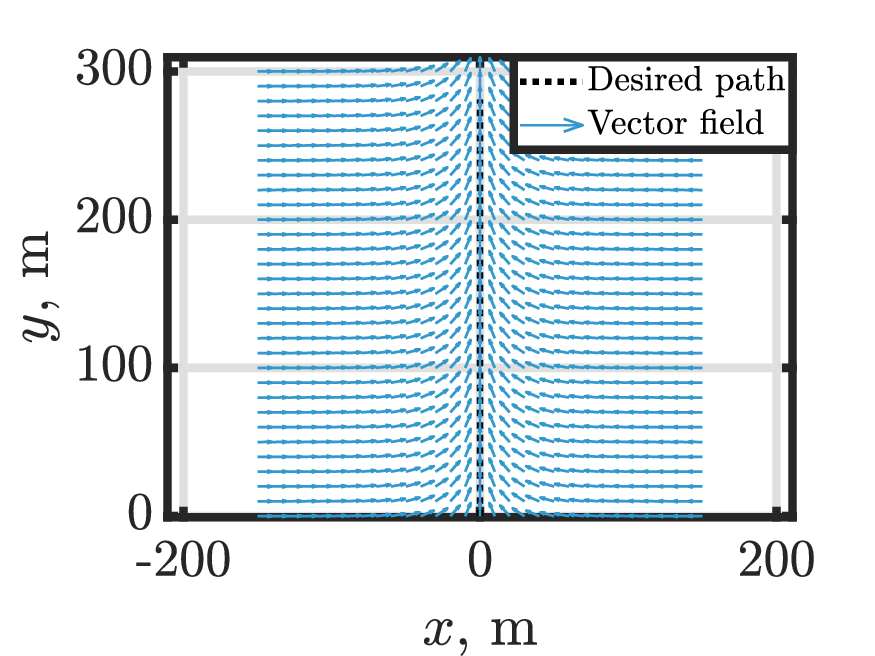}
			\caption{ Vector field}
			\label{fig: Vector field}
		\end{subfigure}%
		\caption{Results for desired guidance command and vector fields}
		\label{fig: Results for course angle and the corresponding vector field}
	\end{figure}
	
	\subsection{Turn-Rate Characteristics}
	
	The commanded turn-rate is computed as
	\begin{equation}
		\dot{\chi}_d
		=
		\frac{d\chi_d}{dx}\dot{x}.
		\label{eq: desired kaidot}
	\end{equation}
	Differentiating \eqref{eq:proposed_chi} with respect to $x$ gives
	\begin{align*}
		\frac{d\chi_d}{dx}
		&=    \frac{k_\mathrm{atanh}\sech^2(k_\mathrm{atanh}x)}
		{\sqrt{1-\tanh^2(k_\mathrm{atanh}x)}}.
	\end{align*}
	Using
	\begin{equation*}
		1-\tanh^2(k_\mathrm{atanh}x)=\sech^2(k_\mathrm{atanh}x),
	\end{equation*}
	we obtain
	\begin{equation}
		\frac{d\chi_d}{dx}
		=
		k_\mathrm{atanh}\sech(k_\mathrm{atanh}x).
		\label{eq:dchidx}
	\end{equation}
	Substituting \eqref{eq:lateral_dynamics} and \eqref{eq:dchidx} in \eqref{eq: desired kaidot} yields
	\begin{equation}
		\dot{\chi}_d
		=
		-k_\mathrm{atanh}V_g\sech(k_\mathrm{atanh}x)\tanh(k_\mathrm{atanh}x).
		\label{eq:turn_rate_expression}
	\end{equation}
	Therefore,
	\begin{equation}
		|\dot{\chi}_d|
		=
		k_\mathrm{atanh}V_g \left|\sech(k_\mathrm{atanh}x)\tanh(k_\mathrm{atanh}x) \right|.
		\label{eq:turn_rate_abs}
	\end{equation}
	Let $z=\tanh(k_\mathrm{atanh}x).$
	Since $z\in(-1,1)$ and $\sech(k_\mathrm{atanh}x)=\sqrt{1-z^2}$, \eqref{eq:turn_rate_abs} becomes
	\begin{equation*}
		|\dot{\chi}_d|
		=
		k_\mathrm{atanh}V_g |z|\sqrt{1-z^2}.
	\end{equation*}
	The maximum value of $|z|\sqrt{1-z^2}$ is obtained by maximizing
	\begin{equation}
		g(z)=z\sqrt{1-z^2}, \qquad z\in[0,1].
		\label{eq: gfun}
	\end{equation}
	Equating the derivative of $g^2(z)$ to zero gives
	\begin{equation*}
		\frac{d}{dz}g^2(z)
		=
		2z(1-2z^2)=0 \Rightarrow z=0
		\quad \text{or} \quad
		z^2=\frac{1}{2}.
	\end{equation*}
	Using $z = \frac{1}{2} $ in \eqref{eq: gfun} gives the maximum value as $g_{\max}=
	\frac{1}{2}$.
	Hence, the maximum commanded turn-rate is
	\begin{equation}
		|\dot{\chi}_d|_{\max}
		=
		\frac{k_\mathrm{atanh}V_g}{2}.
		\label{eq:max_turn_rate}
	\end{equation}
	To satisfy the UAV turn-rate limitation
	\begin{equation*}
		|\dot{\chi}_d| \leq \omega_{\max},
	\end{equation*}
	the shaping gain must satisfy
	\begin{equation}
		k_\mathrm{atanh} \leq \frac{2\omega_{\max}}{V_g}.
		\label{eq:gain_condition}
	\end{equation}
	Equation~\eqref{eq:gain_condition} provides an explicit analytical relation between the vector-field shaping gain and the vehicle turn-rate limit. Fig.~\ref{fig: Commanded turn-rate profiles with lateral error} shows the commanded turn-rate profiles for different shaping gains, highlighting the direct influence of $k_{\mathrm{atanh}}$ on the aggressive behaviour and peak turn-rate magnitude.
	
	\begin{figure}
		\centering
		\includegraphics[width=0.9\linewidth]{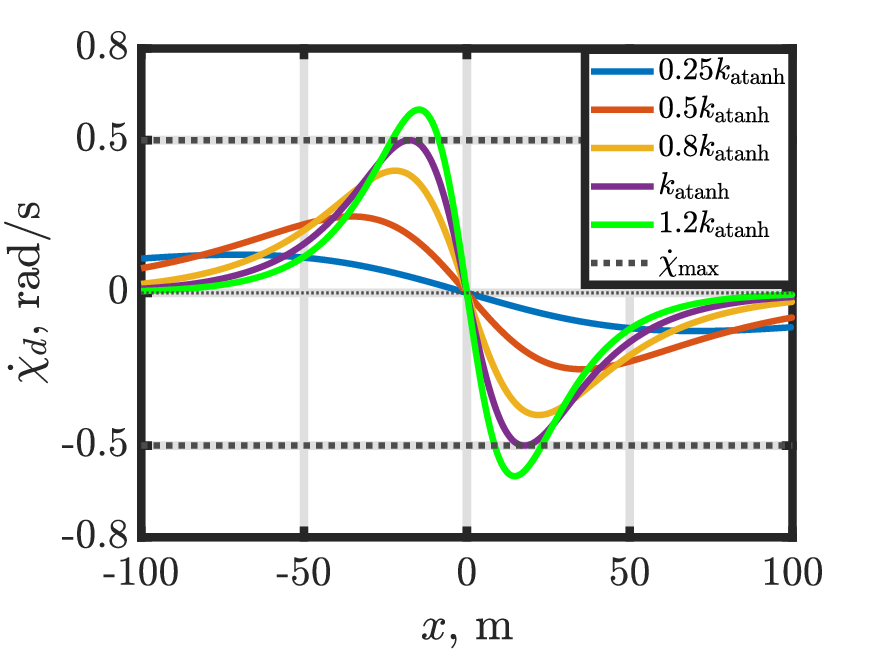}
		\caption{Commanded turn-rate profiles with lateral error}
		\label{fig: Commanded turn-rate profiles with lateral error}
	\end{figure}
	
	\subsection{Stability Properties}
	
	The asymptotic convergence property of the proposed guidance method is discussed in the subsequent theorem.
	
	\begin{theorem}
		Consider the UAV lateral error dynamics governed by $\dot{x} = V_g \cos (\chi)$ under the proposed direct vector-field guidance law \eqref{eq:proposed_chi}. For any $k_\mathrm{atanh}>0$ and constant ground speed $V_g>0$, the desired path $x=0$ is globally asymptotically stable under ideal course $\chi = \chi_d$ following. Furthermore, if the shaping gain satisfies \eqref{eq:gain_condition}
		then the commanded turn-rate satisfies $|\dot{\chi}_d| \leq \omega_{\max}$
		for all $x\in\mathbb{R}$.
	\end{theorem}
	
	\begin{proof}
		Using the proposed guidance law \eqref{eq:proposed_chi}, the lateral error dynamics are given by
		\begin{equation*}
			\dot{x}=-V_g\tanh(k_\mathrm{atanh}x).
		\end{equation*}
		Consider the Lyapunov function as
		\begin{equation*}
			W(x)=\frac{1}{2}x^2.
		\end{equation*}
		Differentiating along the system trajectories gives
		\begin{equation*}
			\dot{W}(x)=x\dot{x}.
		\end{equation*}
		Substituting the lateral dynamics yields
		\begin{equation*}
			\dot{W}(x)
			=
			-V_g x\tanh(k_\mathrm{atanh}x).
		\end{equation*}
		Since $k_\mathrm{atanh}>0$, the function $\tanh(k_\mathrm{atanh}x)$ has the same sign as $x$. Therefore,
		\begin{equation*}
			x\tanh(k_\mathrm{atanh}x)>0, \qquad \forall x\neq 0.
		\end{equation*}
		Thus,
		\begin{equation*}
			\dot{W}(x)<0, \qquad \forall x\neq 0,
		\end{equation*}
		and $\dot{W}(0)=0$. Hence, $x=0$ is globally asymptotically stable.
		Next, from \eqref{eq:max_turn_rate}, we have
		\begin{equation*}
			|\dot{\chi}_d|_{\max}
			=
			\frac{k_\mathrm{atanh}V_g}{2}.
		\end{equation*}
		Therefore, if
		\begin{equation*}
			k_\mathrm{atanh} \leq \frac{2\omega_{\max}}{V_g} \Rightarrow |\dot{\chi}_d|_{\max}
			\leq
			\omega_{\max}.
		\end{equation*}
		Consequently,
		\begin{equation*}
			|\dot{\chi}_d| \leq \omega_{\max},
			\qquad \forall x\in\mathbb{R}.
		\end{equation*}
		This proves both global asymptotic convergence to the path and turn-rate feasibility of the proposed vector field.
	\end{proof}
	
	\begin{remark}
		Increasing $k_\mathrm{atanh}$ produces faster convergence toward the path, but also increases the maximum commanded turn-rate. Therefore, \eqref{eq:gain_condition} gives the largest feasible shaping gain that preserves turn-rate feasibility.
	\end{remark}

	\subsection{Comparative Analytical Discussion}
	
	This section presents a comparative analytical discussion of widely used direct vector-field guidance laws for turn-limited path following. 
	
	\subsubsection{Arctan Vector Field}
	The arctangent-based formulation developed in \cite{nelson2007vector},
	\begin{equation*}
		\chi_d^{\mathrm{atan}}
		=
		\frac{\pi}{2}
		+
		\tan^{-1}(k_\mathrm{atan}x).
	\end{equation*}
	The corresponding lateral dynamics are
	\begin{equation*}
		\dot{x}
		=
		-V_g
		\frac{k_\mathrm{atan}x}{\sqrt{1+k_\mathrm{atan}^2x^2}}.
	\end{equation*}
	The  turn-rate profile is obtained as
	\begin{equation*}
		\dot{\chi}_d^{\mathrm{atan}}
		=
		-\frac{k_\mathrm{atan}^2V_gx}
		{(1+k_\mathrm{atan}^2x^2)^{3/2}}.
	\end{equation*}
	The maximum turn-rate occurs at
	\begin{equation*}
		x=\pm\frac{1}{\sqrt{2}k_\mathrm{atan}},
	\end{equation*}
	yielding
	\begin{equation*}
		\left|\dot{\chi}_d^{\mathrm{atan}} \right|_{\max}
		=
		\frac{2k_\mathrm{atan}V_g}{3\sqrt{3}}.
	\end{equation*}
	
	\subsubsection{Arcsine Vector Field}
	
	The arcsine vector-field guidance law developed in \cite{shivamRatnoo2023} is given by
	\begin{equation*}
		\chi_d^{\mathrm{asin}}
		=
		\frac{\pi}{2}
		+
		\sin^{-1}
		\left(
		\frac{1}{1+k_\mathrm{asin}x^2}
		-1
		\right).
	\end{equation*}
	
	The corresponding turn-rate profile is
	\begin{equation*}
		\dot{\chi}_d^{\mathrm{asin}}
		=
		-\frac{2k_\mathrm{asin}V_gx}
		{(1+k_\mathrm{asin}x^2)^2}.
	\end{equation*}
	
	The maximum turn-rate is analytically obtained as
	\begin{equation*}
		\left| \dot{\chi}_d^{\mathrm{asin}} \right|_{\max}
		=
		\frac{9V_g}{8}
		\sqrt{\frac{k_\mathrm{asin}}{3}}.
	\end{equation*}
	Compared to the arctangent vector field, the arcsine formulation produces lower peak turn-rate requirements for comparable convergence characteristics when starting with the same initial heading directions.

	
	\subsubsection{Arccos Vector Field}
	The Arccos vector-field guidance law is used in \cite{pal2026curvature}
	\begin{equation}
		\chi_d^{\mathrm{acos}} =
		\begin{cases}
			\cos^{-1}\left(1-e^{k_\mathrm{acos} x}\right), & x \leq 0,\\[2mm]
			\pi-\cos^{-1}\left(1-e^{-k_\mathrm{acos} x}\right), & x>0,
		\end{cases}
		\label{eq:acos_vf}
	\end{equation}
	where $k_\mathrm{acos}>0$ is the shaping gain. 
	Following similar lateral dynamics trajectories and solving, one can obtain  the maximum commanded turn-rate of the Arccos vector field is
	\begin{equation}     
		\left| \dot{\chi}_d^{\text{acos}} \right|_{\max}
		=
		V_g k_\mathrm{acos}
		\sqrt{\frac{5\sqrt{5}-11}{2}} \approx 0.3003 V_g k_\text{acos}.
		\label{eq:acos_max_turn_rate}
	\end{equation}
	The shaping gain must satisfy
	\begin{equation}   
		k_\mathrm{acos}
		\leq
		\frac{\omega_{\max}}
		{V_g}
		\sqrt{\frac{2}{5\sqrt{5}-11}} \lesssim
		\frac{\omega_{\max}}
		{0.3003 V_g}.
		\label{eq:acos_gain_condition}
	\end{equation}
	Table~\ref{tab: qualitative comparison} summarizes the qualitative comparison using parameter $\rho=\omega_{\max}/V_g$ among the considered direct vector-field guidance laws.
	\begin{table*}[!t]
		\centering
		\caption{Comparative turn-rate characteristics and feasible gain selection using $\rho=\omega_{\max}/V_g$.}
		\label{tab: qualitative comparison}
		\renewcommand{\arraystretch}{1.25}
		\begin{tabular}{llll}
			\toprule
			\multicolumn{1}{c}{\textbf{Vector Field}} &
			\multicolumn{1}{c}{\textbf{Maximum Turn-rate}} &
			\multicolumn{1}{c}{\textbf{Feasible Gain}} &
			\multicolumn{1}{c}{\textbf{Remark}} \\
			\midrule
			
			Arctan VF &
			$\displaystyle \left| \dot{\chi}_d \right|_{\max}=\frac{2k_\mathrm{atan}V_g}{3\sqrt{3}}$
			&
			$\displaystyle k_\mathrm{atan} \leq 2.598\,\rho$
			&
			Smooth but gradual variation.
			\\[4mm]
			
			Arcsine VF &
			$\displaystyle |\dot{\chi}_d|_{\max}=\frac{9V_g}{8}\sqrt{\frac{k_\mathrm{asin}}{3}}$
			&
			$\displaystyle k_\mathrm{asin} \leq 2.370\,\rho^2$
			&
			Strong local steering; gain has unit $\rm m^{-2}$.
			\\[4mm]
			
			Arccos VF &
			$\displaystyle |\dot{\chi}_d|_{\max}=k_\mathrm{acos}V_g\sqrt{\frac{5\sqrt{5}-11}{2}}$
			&
			$\displaystyle k_\mathrm{acos} \leq 3.330\,\rho$
			&
			Strong far-field attraction.
			\\[4mm]
			
			Proposed VF &
			$\displaystyle |\dot{\chi}_d|_{\max}=\frac{k_\mathrm{atanh}V_g}{2}$
			&
			$\displaystyle k_\mathrm{atanh} \leq 2.000\,\rho$
			&
			Compact and smooth shaping profile.
			\\[2mm]
			
			\bottomrule
		\end{tabular}
	\end{table*}
	
	\begin{remark}
		The ratio $\rho=\omega_{\max}/V_g$ represents the available maneuverability of the UAV.
		However, the numerical coefficients in Table~\ref{tab: qualitative comparison} should not be interpreted as a convergence speed, because each vector field has a different nonlinear shaping profile and, in the arcsine case, the gain has a different dimension. 
	\end{remark}

	The above comparison highlights that the proposed   vector field combines analytical tractability and smooth convergence behavior while providing an explicit turn-feasibility-based gain-selection mechanism suitable for turn-limited UAV path-following applications.
	
	\section{Simulation Results}
	
	This section discusses a straight path-following scenario. 
	\subsection{Sample scenarios with ideal initial heading}
	To highlight the role of the analytical gain-selection condition, two cases are considered: a turn-feasible gain satisfying the proposed bound and an aggressive gain that violates the bound.
	The simulations were performed with a ground speed of $V_g=20~\mathrm{m\,s^{-1}}$, a prescribed turn-rate limit of $\omega_{\max}=0.5~\mathrm{rad\,s^{-1}}$, and an initial lateral error of $x_0=160~\mathrm{m}$. Two gain values were considered: a turn-rate-feasible gain $k_\mathrm{atanh}=0.0400$, selected according to the analytical condition $k_\mathrm{atanh}\leq 2\omega_{\max}/V_g$, and an aggressive gain $k_\mathrm{atanh}=0.1250$, chosen beyond the feasible limit. Fig.~\ref{fig: Results for following straight line path} illustrates the effect of guidance-gain selection on the path-following performance. As shown in Fig.~\ref{fig: UAV trajectories ideal}-\ref{fig: Error profiles ideal}, both gains successfully drive the UAV toward the desired path and exhibit similar convergence properties. However, the larger gain produces a steeper course-angle variation, as observed in Fig.~\ref{fig: Commanded course angle variations}-\ref{fig: Commanded course angle variations with lateral error}, resulting in a more aggressive steering response. The corresponding turn-rate profiles shown in Fig.~\ref{fig: Turn rate variations}-\ref{fig: Turn-rate variations with lateral error} depict that the feasible gain remains entirely within the prescribed turn-rate envelope, with an analytical peak turn-rate of $0.40~\mathrm{rad\,s^{-1}}$, whereas the aggressive gain generates a peak turn-rate of $1.25~\mathrm{rad\,s^{-1}}$, exceeding the allowable limit. These results validate the derived gain-selection condition and demonstrate that the proposed analytical design framework can guarantee maneuvering-feasible path following without significantly compromising convergence performance.
	\begin{figure}
		\begin{subfigure}[b]{.25\textwidth}
			\centering			\includegraphics[trim={2pt 3pt 2pt 3pt}, clip, width=\linewidth]{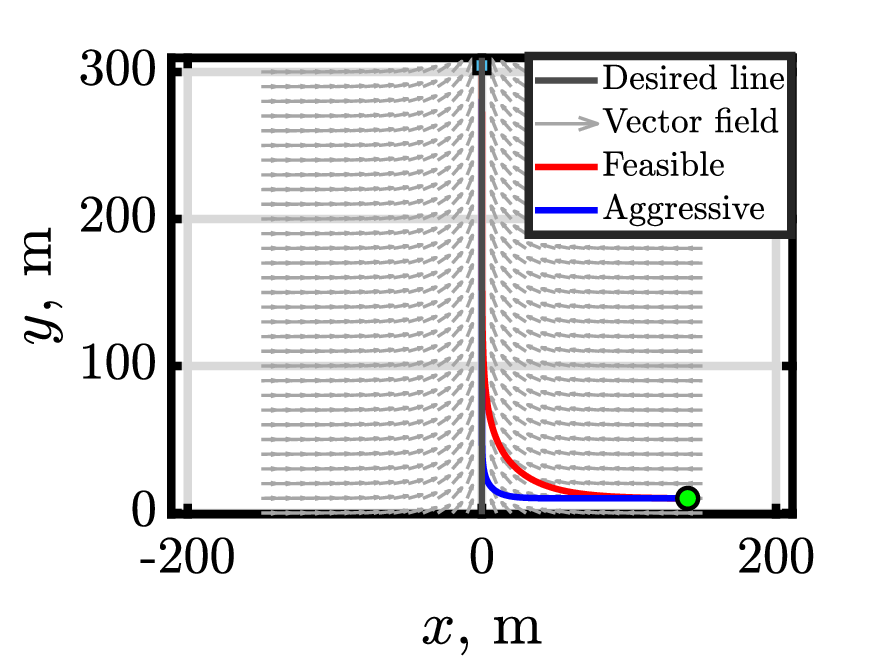}
			\caption{ UAV trajectory}
			\label{fig: UAV trajectories ideal}     
		\end{subfigure}%
		\begin{subfigure}[b]{.25\textwidth}
			\centering		\includegraphics[trim={2pt 3pt 2pt 3pt}, clip, width=\linewidth]{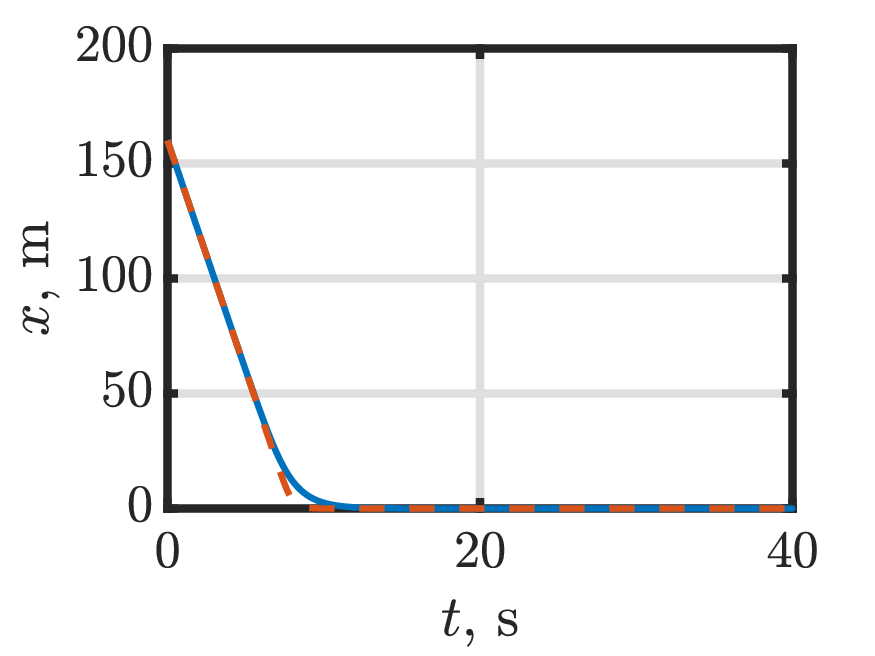}
			\caption{ Tracking error profiles}
			\label{fig: Error profiles ideal}
		\end{subfigure}%
		\qquad
		\begin{subfigure}[b]{.25\textwidth}
			\centering			\includegraphics[trim={2pt 3pt 2pt 3pt}, clip, width=\linewidth]{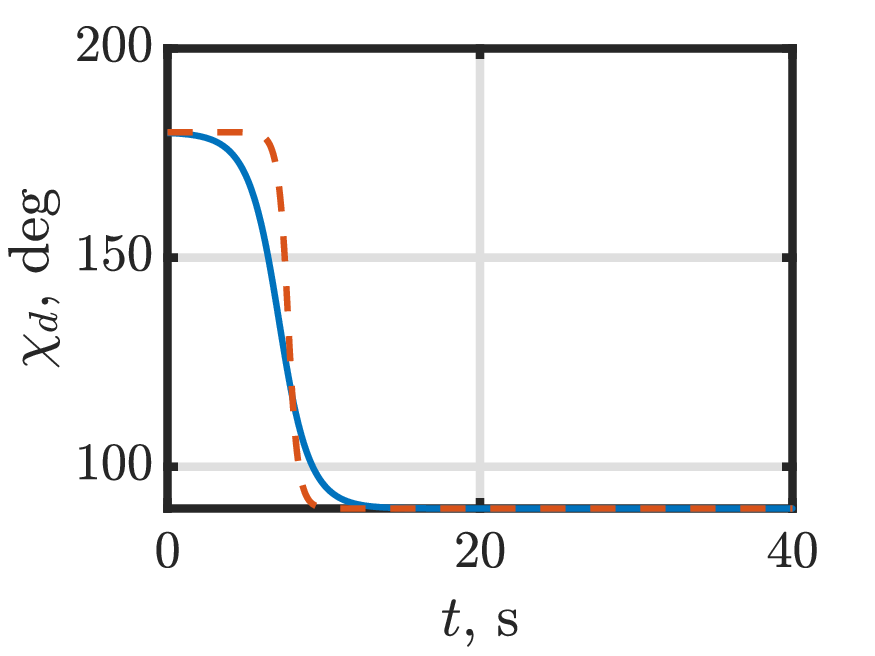}
			\caption{ Course angle variations}
			\label{fig: Commanded course angle variations}     
		\end{subfigure}%
		\begin{subfigure}[b]{.25\textwidth}
			\centering		
			\includegraphics[trim={2pt 3pt 2pt 3pt}, clip,width=\linewidth]{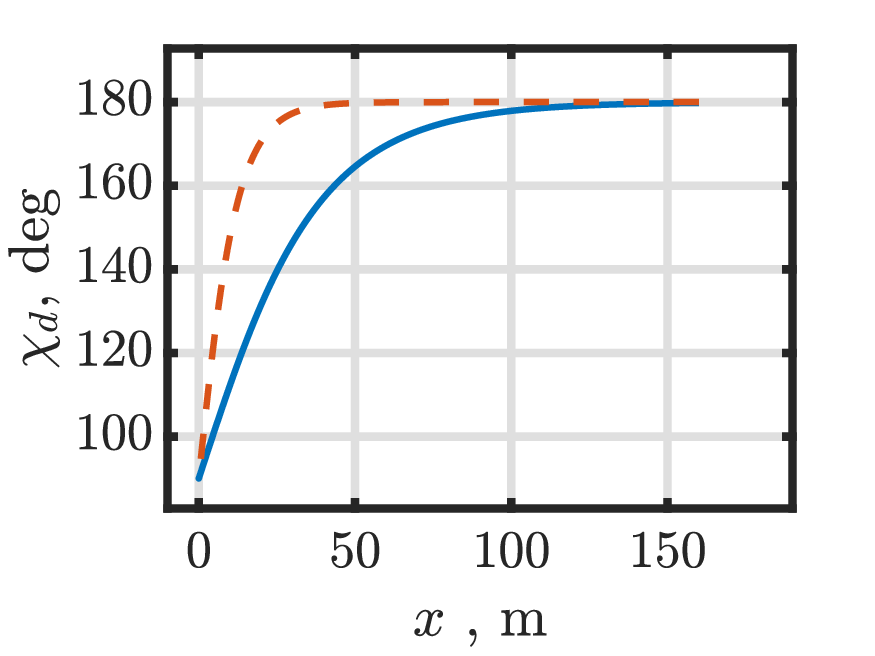}
			\caption{ Course angle with lateral error }
			\label{fig: Commanded course angle variations with lateral error}
		\end{subfigure}%
		\qquad
		\begin{subfigure}[b]{.25\textwidth}
			\centering
			\includegraphics[trim={2pt 3pt 2pt 3pt}, clip, width=\linewidth]{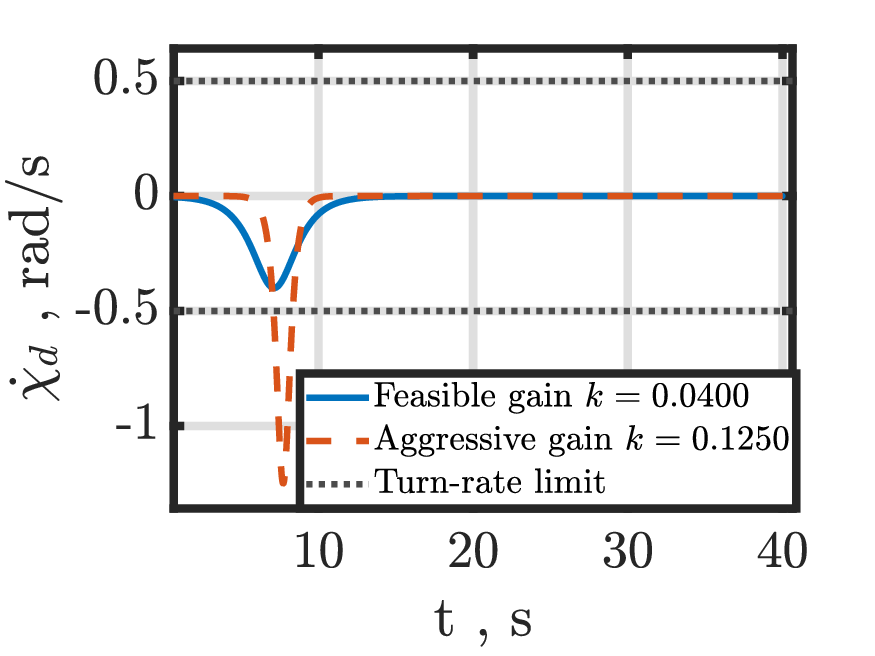}
			\caption{ Turn-rate variations}
			\label{fig: Turn rate variations}     
		\end{subfigure}%
		\begin{subfigure}[b]{.25\textwidth}
			\centering		\includegraphics[trim={2pt 3pt 2pt 3pt}, clip,width=\linewidth]{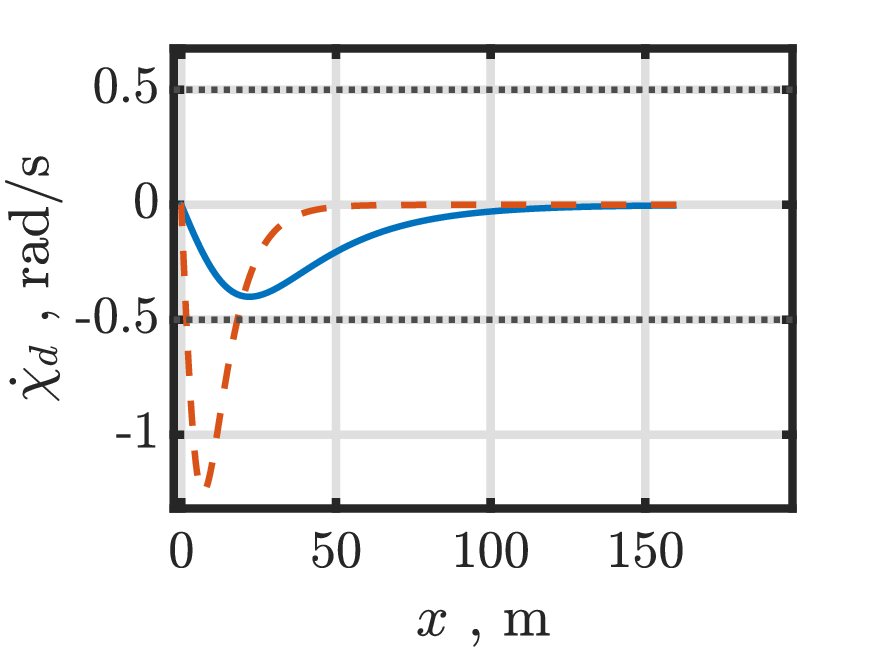}
			\caption{ Turn-rate  with lateral error}
			\label{fig: Turn-rate variations with lateral error}
		\end{subfigure}%
		\caption{ Results for following straight line path}
		\label{fig: Results for following straight line path}
	\end{figure}
	\subsection{Scenarios with non-ideal initial heading}
	Three initialisation scenarios are considered: ideal heading alignment, a positive heading offset of $90^\circ$, and a negative heading offset of $60^\circ$.
	Fig.~\ref{fig: Results for following straight line path nonideal} investigates the effect of non-ideal initial course-angle conditions using the same simulation parameters as Case~1. As shown in Fig.~\ref{fig: UAV trajectories nonideal}, all trajectories successfully converge toward the desired path despite the large initial heading deviations. The lateral-error responses in Fig.~\ref{fig: Error profiles nonideal} exhibit similar convergence characteristics, indicating that the proposed vector field effectively recovers from significant initialization errors. The course-angle error profiles in Fig.~\ref{fig: Course angle error profiles nonideal} show a transient heading-correction phase followed by convergence to the desired vector-field direction. Furthermore, the turn-rate profiles in Fig.~\ref{fig: Turn-rate profiles} remain within the prescribed maneuverability envelope throughout the maneuver. Figs.~\ref{fig: Commanded course angle variations with lateral error nonideal} and \ref{fig: Turn-rate variations with lateral error nonideal} show the course-angle and turn-rate responses as functions of lateral error for different initial heading conditions. Although the trajectories begin with substantial course mismatches, they converge toward the same vector-field profile without violating the prescribed turn-rate bound, demonstrating robustness to nonideal initialization.
	
	\begin{figure}
		\begin{subfigure}[b]{.25\textwidth}
			\centering			\includegraphics[trim={2pt 3pt 2pt 3pt}, clip, width=\linewidth]{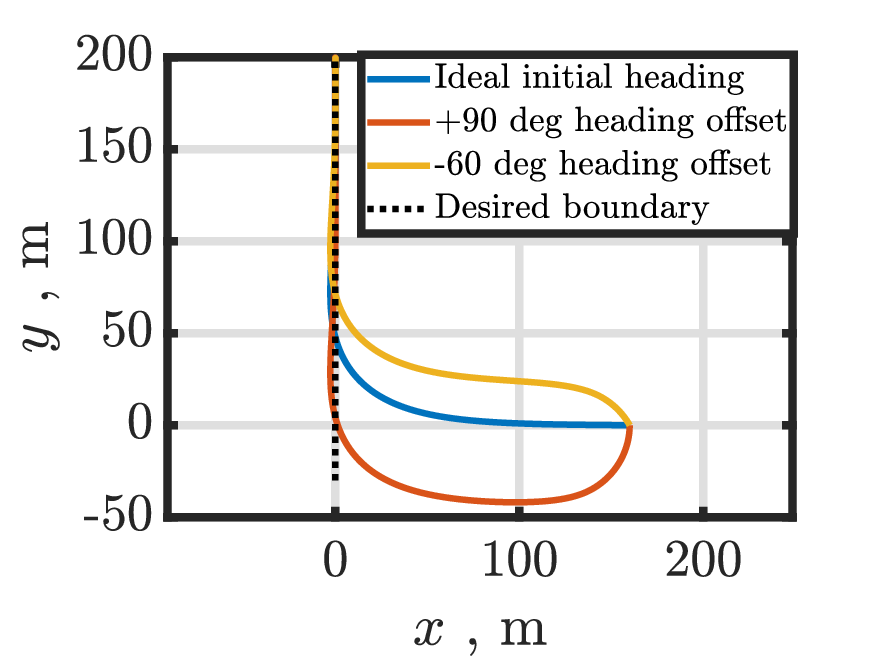 }
			\caption{ UAV trajectories}
			\label{fig: UAV trajectories nonideal}     
		\end{subfigure}%
		\begin{subfigure}[b]{.25\textwidth}
			\centering		\includegraphics[trim={2pt 3pt 2pt 3pt}, clip, width=\linewidth]{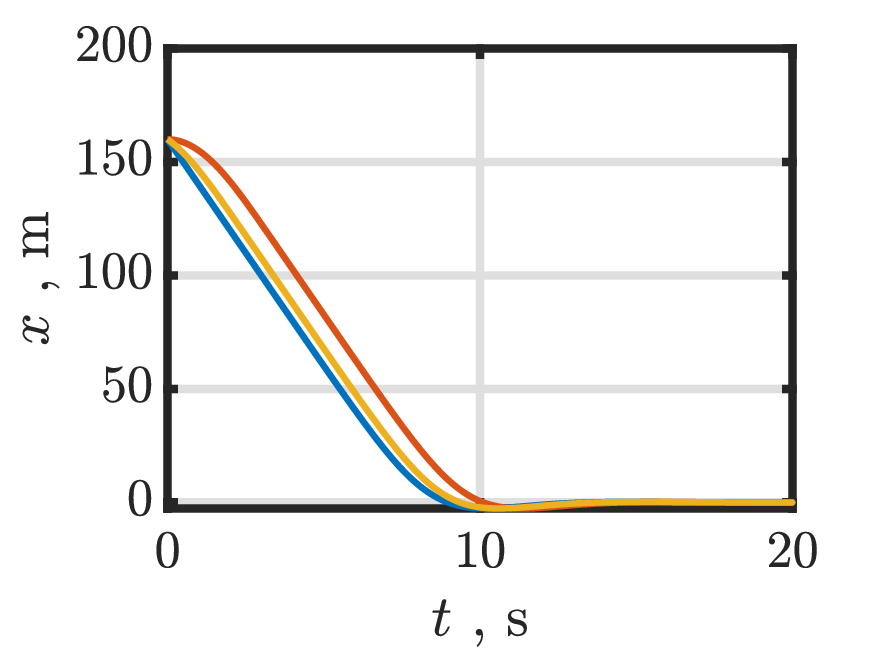}
			\caption{ Tracking error profiles}
			\label{fig: Error profiles nonideal}
		\end{subfigure}%
		\qquad
		\begin{subfigure}[b]{.25\textwidth}
			\centering			\includegraphics[trim={2pt 3pt 2pt 3pt}, clip, width=\linewidth]{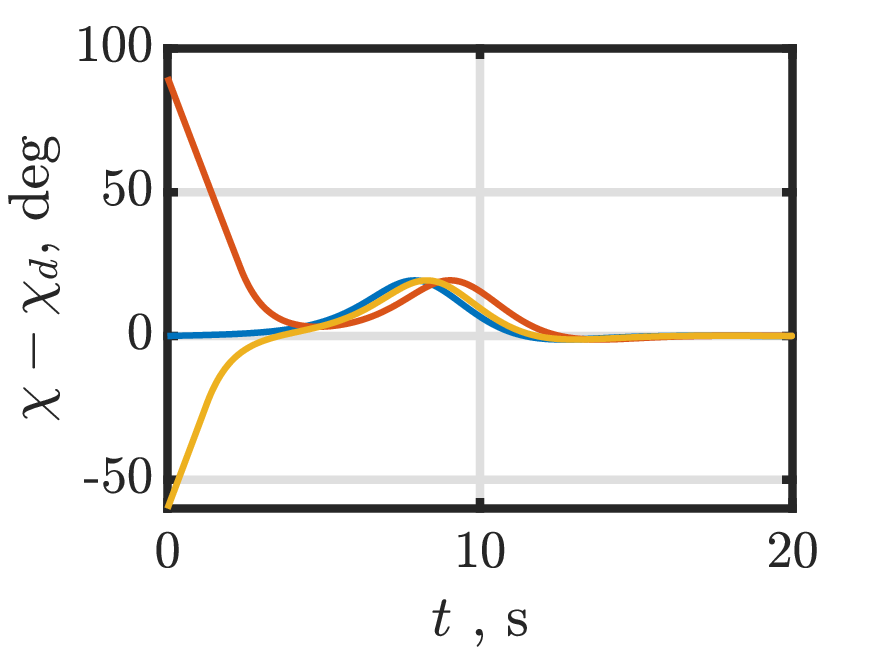}
			\caption{ Course angle error profiles}
			\label{fig: Course angle error profiles nonideal}     
		\end{subfigure}%
		\begin{subfigure}[b]{.25\textwidth}
			\centering		\includegraphics[trim={2pt 3pt 2pt 3pt}, clip,width=\linewidth]{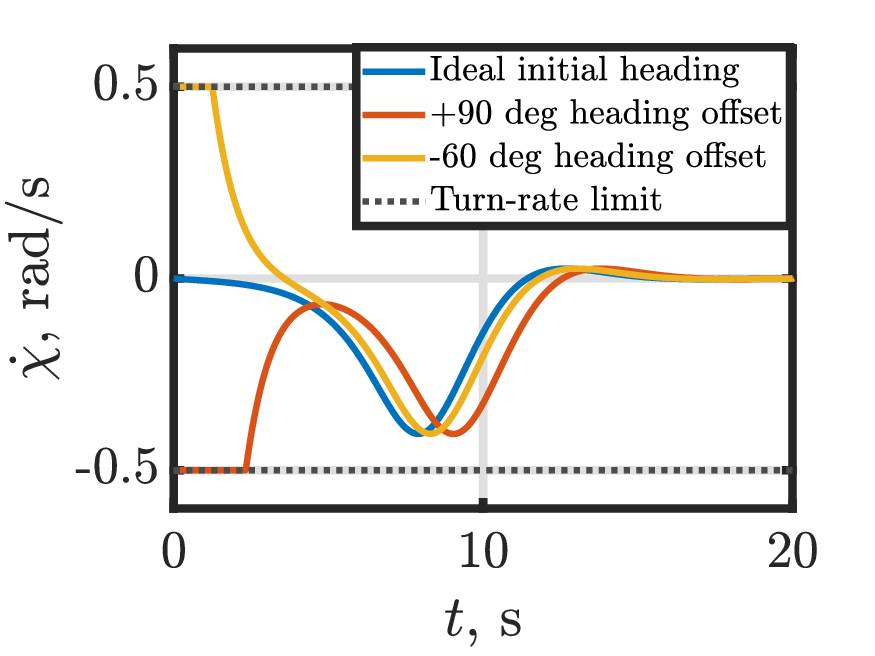}
			\caption{Turn-rate profiles}
			\label{fig: Turn-rate profiles}
		\end{subfigure}%
		\qquad 
		\begin{subfigure}[b]{.25\textwidth}
			\centering		
			\includegraphics[trim={2pt 3pt 2pt 3pt}, clip,width=\linewidth]{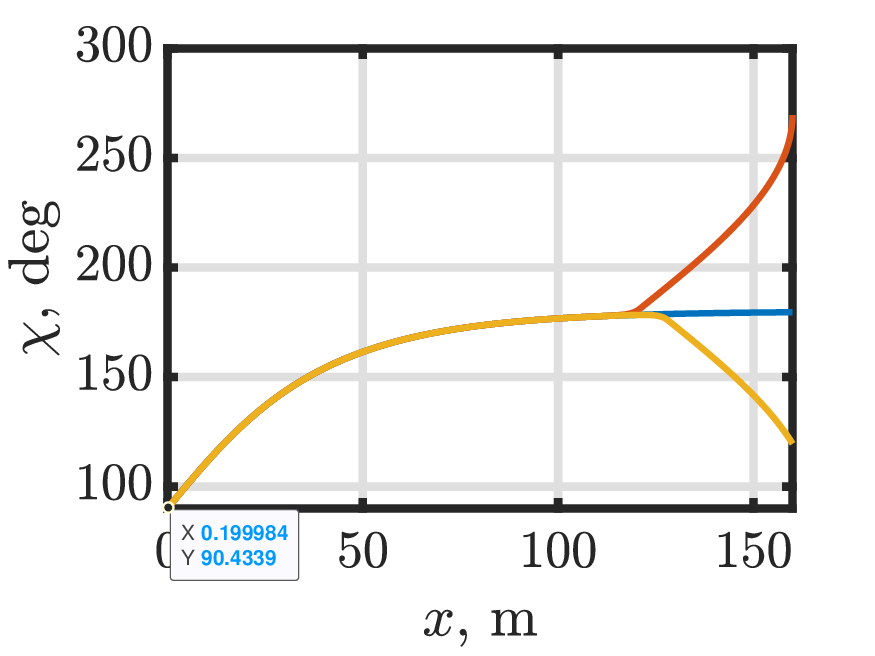}
			\caption{ Course angle with lateral error }
			\label{fig: Commanded course angle variations with lateral error nonideal}
		\end{subfigure}%
		\begin{subfigure}[b]{.25\textwidth}
			\centering		\includegraphics[trim={2pt 3pt 2pt 3pt}, clip,width=\linewidth]{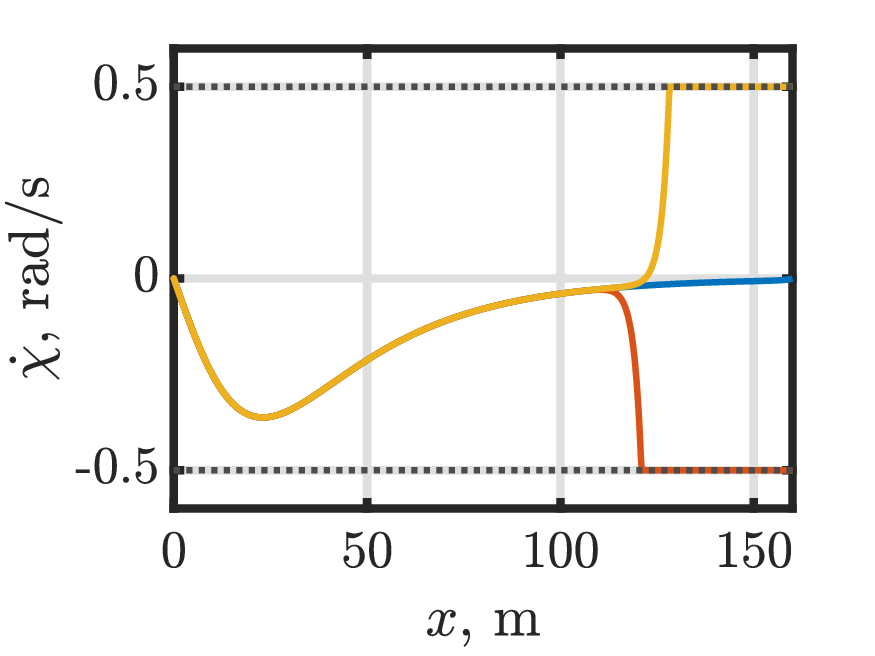}
			\caption{ Turn-rate  with lateral error}
			\label{fig: Turn-rate variations with lateral error nonideal}
		\end{subfigure}%
		\caption{Results for non-ideal initial course direction}
		\label{fig: Results for following straight line path nonideal}
	\end{figure}
	
	\subsection{Comparative study}
	This case investigates the influence of the maneuverability ratio
	$\rho=\omega_{\max}/V_g$ on the steering command required by different
	vector-field guidance laws. The simulations were performed with a
	constant ground speed of $V_g=20~\mathrm{m\,s^{-1}}$ and three representative
	fixed-wing maneuverability ratios, namely
	$\rho \in \{0.015,\,0.025,\,0.035\}$, corresponding to maximum allowable
	turn-rates of $\omega_{\max} \in \{0.3,\,0.5,\,0.7\}~\mathrm{rad\,s^{-1}}$.
	For each guidance law, the gain was selected according to the
	analytically derived feasibility condition is summarized in
	Table~\ref{tab: qualitative comparison}, ensuring that the commanded Turn-rate
	remained compatible with the prescribed maneuvering envelope.
	\begin{figure}
		\centering   \includegraphics[width=0.9\linewidth]{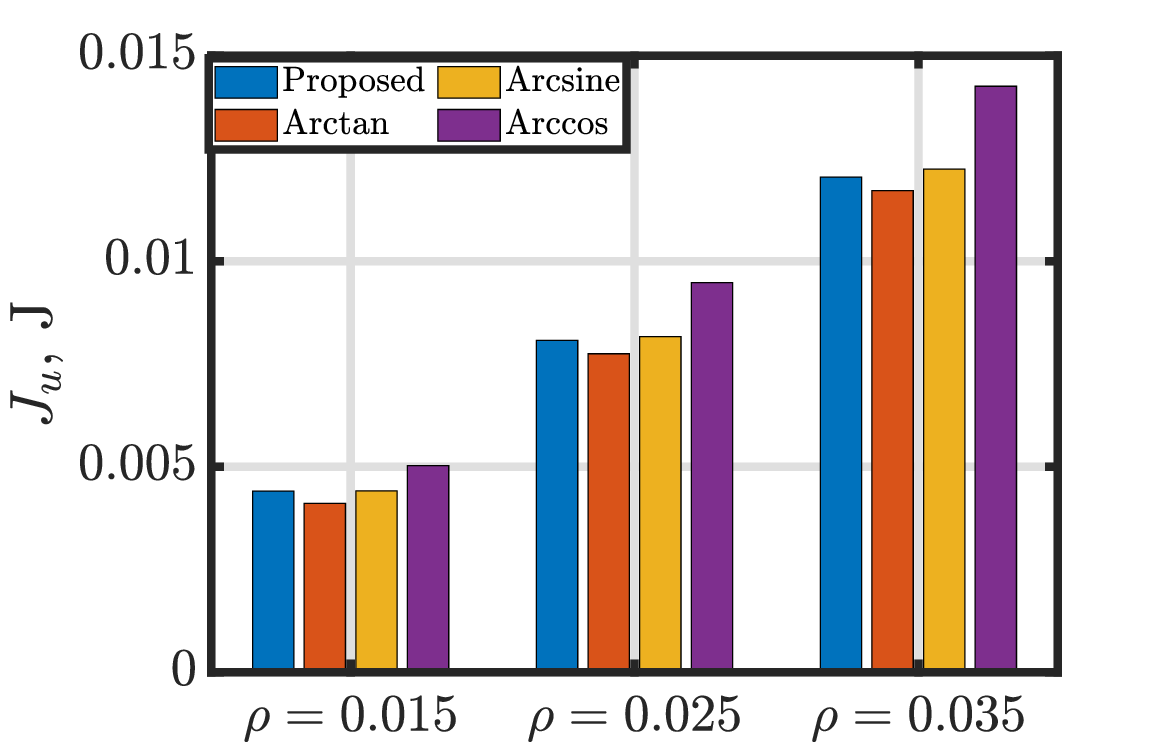}
		\caption{Comparative average control effort analysis}
		\label{fig: Comparative average control effort analysis}
	\end{figure}
	Fig.~\ref{fig: Comparative average control effort analysis} presents the average control effort
	$J_u=(1/T)\int_0^T \dot{\chi}^2\,dt$ for each of the vector field methods. As expected, the average control
	effort increases with $\rho$ for all guidance laws because larger
	turn-rate limits lead to more aggressive steering actions and faster
	boundary convergence. Among the methods, the arccos vector field
	consistently exhibits the highest control effort owing to its strong
	far-field attraction characteristics. In contrast, the
	arctan vector field produces the lowest steering effort. The
	proposed vector field lies between these two extremes, requiring only a
	moderate increase in steering effort while preserving a smooth sigmoid
	variation and a compact gain-selection rule. The arcsine vector field
	shows a similar trend but tends to demand slightly higher steering
	commands due to its stronger local attraction
	properties.
	Overall, the results indicate that the proposed vector field achieves a
	balanced compromise between convergence aggressiveness and control
	effort. Rather than minimizing control effort alone, the proposed
	formulation maintains turn-rate feasibility through an explicit
	analytical gain design while avoiding the excessive steering demand
	associated with highly aggressive vector-field structures.
	The proposed vector field exhibits a smoother convergence profile compared with the more aggressive arcsine and arccosine fields. Consequently, its settling time can be slightly larger under the same maneuverability ratio. This behavior is consistent with the normalized sigmoid shaping structure, which avoids abrupt heading variations and distributes the steering demand more gradually. Therefore, the proposed formulation is better interpreted as a smooth turn-rate-feasible vector field guidance.
	\vspace{-0.5cm}
	\section{Conclusion}
	
	This paper proposes a direct vector-field guidance law for turn-rate-limited path-following. The guidance strategy uses a tanh-hyperbolic sigmoid course-angle profile, resulting in a smooth vector field that continuously drives the vehicle toward the desired path.
	A closed-form expression for the maximum turn-rate was derived, enabling the formulation of an explicit gain-selection condition that guarantees compliance with prescribed turn-rate constraints.
	The stability and convergence properties of the proposed guidance law were analyzed through the lateral-error dynamics, and numerical simulations were conducted under both ideal and non-ideal initial heading conditions.
	Simulation results illustrate path-following convergence while respecting the prescribed turn-rate envelope. A comparative study involving arctan-, arcsine-, and arccos-based vector-field guidance laws further highlighted the influence of vector-field shaping on steering demand and convergence behavior. While the arctan and arccos formulations represent smoother and more aggressive extremes, respectively, the proposed guidance law provides a balanced compromise between steering smoothness, convergence characteristics, and analytical tractability. In particular, the derived gain-selection rule offers a compact and intuitive mechanism for enforcing turn-rate feasibility across different maneuverability regimes.  
	Future work will focus on extending the proposed framework to curved boundaries, closed-orbit surveillance missions, wind-perturbed operations, and high-fidelity six-degree-of-freedom UAV models. The analytical gain-design methodology also provides a foundation for future developments involving curvature-aware and geometry-constrained vector-field guidance strategies.
	
	\section*{Declaration of Usage of Generative ai}
	{
		The authors used generative AI tools to improve the readability of the final draft. In the end, the authors reviewed and edited the content as needed and take full responsibility for the content of the publication.
	}
	
	
	\bibliographystyle{ieeetr}
	\bibliography{ref_CONTROLO}
	
\end{document}